%% file: main.tex
\documentclass[conference,twocolumn, twoside]{IEEEtran}
\usepackage{amssymb,epsfig,graphics, graphicx, url,times,mathrsfs,algorithm,algorithmic,amsmath,array,latexsym,fancyhdr,xspace,wrapfig, xcolor,multirow,amsthm,caption,cite,helvet,sidecap,bm,comment,upref,bbm}
\usepackage[normalem]{ulem}
\usepackage{balance}

\usepackage[normalem]{ulem}

\newcommand{\rvY}{\texttt{Y}}

\usepackage[inline]{enumitem}
\usepackage{arydshln}
\usepackage[makeroom]{cancel}
\usepackage{varwidth}
\usepackage{subcaption}
\usepackage{pgfplots}
\usepackage{nicefrac}

\usepackage{tikz}
\usetikzlibrary{shapes.geometric}
\usetikzlibrary{shapes,snakes,patterns}
\usetikzlibrary{arrows,positioning,automata,calc}
\usetikzlibrary{plotmarks}

\input{notations}

\newcommand{\transpose}{\mathsf{T}}

\IEEEoverridecommandlockouts

\begin{document}

\title{{Distributed Matrix-Vector Multiplication with Sparsity and Privacy Guarantees}}
\vspace{-0.4cm}
\author{\IEEEauthorblockN{
Marvin Xhemrishi, Rawad Bitar, and Antonia Wachter-Zeh }\\ 
\vspace{-0.48cm}
\IEEEauthorblockA{Institute for Communications Engineering, Technical University of Munich, Munich, Germany\\
\{\texttt{marvin.xhemrishi, rawad.bitar, antonia.wachter-zeh}\}\texttt{@tum.de}}
\thanks{M.~Xhemrishi's work was funded by the DFG (German Research Foundation) project under Grant Agreement No. WA 3907/7-1. {This work was partly supported by the Technical University of Munich - Institute for Advanced Studies, funded by the German Excellence Initiative
and European Union Seventh Framework Programme under Grant Agreement
No. 291763.}
}
\vspace{-1cm}
}
\vspace{-1cm}

\maketitle
\begin{abstract}
 We consider the problem of designing a coding scheme that allows both sparsity and privacy for distributed matrix-vector multiplication. Perfect information-theoretic privacy requires encoding the input sparse matrices into matrices distributed uniformly at random from the considered alphabet; thus destroying the sparsity. Computing matrix-vector multiplication for sparse matrices is known to be fast. Distributing the computation over the non-sparse encoded matrices maintains privacy, but introduces artificial computing delays. In this work, we relax the privacy constraint and show that a certain level of sparsity can be maintained in the encoded matrices. We consider the {\master}/{\worker} setting while assuming the presence of two clusters of {\workers}: one is completely untrusted in which all {\workers} collude to eavesdrop on the input matrix and in which perfect privacy must be satisfied; in the partly trusted cluster, only up to $z$ {\workers} may collude and to which revealing small amount of information about the input matrix is allowed. We design a scheme that trades sparsity for privacy while achieving the desired constraints. We use cyclic task assignments of the encoded matrices to tolerate partial and full stragglers.
\end{abstract}

\section{Introduction}\label{sec:Intro}
\input{intro.tex}

\section{Preliminaries}\label{sec:Preliminaries}
\input{preliminaries.tex}

\section{Trading Privacy for Sparsity} \label{sec:Schemes_trading_privacy_for_sparsity}
\input{sparse_scheme.tex}

\section{{Coded Computing with Sparsity and Privacy Guarantees}}\label{sec:coded_comp_scheme}
\input{fractional_repetition_scheme.tex}
\section{Conclusion and Future Directions}\label{sec:future_direction}
\input{Future_direction.tex}

\balance
\bibliographystyle{ieeetr}
\bibliography{IEEEabrv,main}

\end{document}

%% file: notations.tex
\definecolor{green1}{rgb}{0,0.5,0}
\definecolor{magenta}{rgb}{1.0, 0.11, 0.81}
\definecolor{mulberry}{rgb}{0.77, 0.29, 0.55}
\definecolor{xgray}{rgb}{0.9, 0.9, 0.9}
\def \blue{\color{blue}}

\def \bes{\begin{equation*}}
\def \ees{\end{equation*}}
\def \bas{\begin{align*}}
\def \eas{\end{align*}}
\def \be{\begin{equation}}
\def \ee{\end{equation}}
\def \bbm{\begin{bmatrix}}
\def \ebm{\end{bmatrix}}

\def \rvA{\texttt{A}}
\def \rvB{\texttt{B}}

\def \rvR {\texttt{R}}

\def \rvT{\texttt{T}}
\def \rvX{\texttt{X}}
\def \rvY{\texttt{Y}}

\def \F{\mathbb{F}}

\newcommand{\Aij}{\rvA_{\{i,j\}}}
\newcommand{\Rij}{\rvR_{\{i,j\}}}
\newcommand{\Xij}[2]{\rvX_{\{#1,#2\}}}

\newcommand{\cZ}{{\cal Z}}

\newcommand{\bfx}{{\boldsymbol x}}
\newcommand{\bfy}{{\boldsymbol y}}

\newcommand{\bfA}{{\mathbf A}}
\newcommand{\bfB}{{\mathbf B}}

\newcommand{\bfR}{{\mathbf R}}

\newcommand{\bfT}{{\mathbf T}}

\newcommand{\bfX}{{\mathbf X}}

\newcommand{\bfAij}{\bfA_{\{i,j\}}}
\newcommand{\bfRij}{\bfR_{\{i,j\}}}

\newcommand{\bfXij}{\bfX_{\{i,j\}}}
\newcommand{\I}{\textrm{I}}
\newcommand{\mutinf}{\I_q}
\newcommand{\entropy}{\textrm{H}_q}
\newcommand{\leakage}[1]{\textrm{L}_#1(p)}
\newcommand{\leakagep}[1]{\textrm{L}_#1(p_{z0},p_{nz0})}

\newcommand{\sparsity}{\textrm{S}}
\newcommand{\pz}{p_{z0}}
\newcommand{\pnz}{p_{nz0}}
\newcommand{\master}{chief}
\newcommand{\worker}{worker}
\newcommand{\Worker}{Worker}
\newcommand{\workers}{workers}

\newtheorem{theorem}{Theorem}

\newtheorem{lemma}{Lemma}
\newtheorem{observation}{Observation}

\newtheorem{definition}{Definition}

\newtheorem{remark}{Remark}

\usetikzlibrary{decorations.pathreplacing,calc}
\tikzset{brace/.style={decorate, decoration={brace}},
 brace mirrored/.style={decorate, decoration={brace,mirror}},
}

\newcounter{brace}

%% file: intro.tex
With the emergence of machine learning applications, the necessity of performing intensive computations is increasing. In several applications, performing the intensive computations on a single processing node is computationally infeasible. Distributed computing arose as a ubiquitous solution. However, distributing the computation comes at the expense of privacy and latency challenges. 

We focus on the {\master}/{\worker} setting in which a main computational node called {\master} wants to run intensive computations on its data. The {\master} divides the intensive computation task into smaller tasks assigned to computation nodes called {\workers}. Waiting for all {\workers} is prone to the presence of \emph{stragglers}, i.e., slow or unresponsive {\workers}~\cite{Deanetal}. Coding-theoretic techniques were proposed as a promising solution to mitigate the effect of straggler, thus speeding up the overall distributed computation \cite{speeding_up_using_codes}. In several applications, the computation is run on sensitive and private data, e.g., medical records and genomes. Leaking information about such data may violate privacy policies and potentially harm the owner of the data \cite{simple_demographics}. Therefore, when distributing the computational tasks to untrusted {\workers}, extra care must be taken to preserve the privacy of the data.

Matrix-vector multiplication is a key computation of many machine learning algorithms, such as principal component analysis, support vector machines and other gradient-descent based algorithms~\cite{suykens1999least,seber2012linear}. Codes mitigating stragglers in distributed matrix-vector multiplication, e.g., \cite{speeding_up_using_codes,Albin,GauriJoshi,Convolutional_codes,li2020coded} destroy the underlying structure of the input matrices. Of particular importance are applications involving the multiplication of \emph{sparse} matrices~\cite{wright2008robust}, i.e., matrices that have a relatively small number of non-zero entries. Sparse matrices are efficiently stored and allow fast and efficient computations \cite{adaptive_sparse_matrix,implementing_sparse_matrix_vector}. Destroying the sparsity of the input matrices may incur artificial delays at the computing nodes. Hence, coding-theoretic techniques maintaining the sparsity of the underlying matrices and mitigating the stragglers are investigated~\cite{coded_sparse_mm,coded_sparse_matrix_leverages_partial_stragglers}. %

We consider information-theoretic privacy, i.e., the eavesdropper has unbounded computational power. When information-theoretic privacy is required, most of the works consider perfect information-theoretic privacy, i.e., no information about the input data is leaked to the eavesdropper, e.g., \cite{Staircase, PRAC, RPM3,Kakar, Burak_private}. In order to achieve perfect privacy, the matrices sent as computational tasks to the {\workers} are padded (mixed) with random matrices generated uniformly at random from a given alphabet. The random matrices are required to be generated independently and uniformly at random to achieve perfect privacy. However, such matrices have a dense structure, i.e., the number of non-zero elements is relatively high, which destroys the sparsity of the computational tasks and increases the time needed to finish the distributed computation, see for example the analysis in \cite{coded_sparse_mm}. %

In this work, we focus on designing a sparse and private matrix-vector multiplication scheme. We relax the perfect privacy constraint. The random matrices are then not required to be generated uniformly at random from the desired alphabet. %

\textit{Related work:} Using codes for straggler mitigation in distributed matrix-vector multiplication obtained a significant interest from the scientific community, e.g., \cite{Bivarite_non_private, Anton_Alex, Albin, Convolutional_codes, FLT_codes, GauriJoshi, MatDot}. The works in \cite{PRAC, Staircase}, consider one-sided privacy, i.e., the matrix-vector multiplication setting in which the input matrix must remain private and the vector can be revealed to the {\workers}. On the other hand, the works in \cite{RPM3, GASP,Kakar, Burak_private, On_the_capacity, batch,LCC} consider double-sided privacy, i.e., the setting of matrix-matrix multiplication in which both input matrices must be kept private. The works of \cite{coded_sparse_mm} and \cite{coded_sparse_matrix_leverages_partial_stragglers} are among the first that consider a sparsity-preserving coded computing scenario. In \cite{coded_sparse_mm}, the authors consider a matrix-matrix multiplication where both input matrices are sparse. The matrices are encoded using Fountain codes \cite{LT_codes} with a  custom-made degree distribution to ensure sparsity. The authors of \cite{coded_sparse_matrix_leverages_partial_stragglers} consider, among other scenarios, a distributed matrix-vector multiplication setting in which the input data is partitioned into smaller matrices. Those matrices are then distributed to the processing units using a {fractional repetition code} to mitigate stragglers. This technique ensures sparsity of the assigned tasks and tolerates stragglers. To increase the straggler tolerance, the authors propose an additional layer of non-sparse coded matrices distributed to the {\workers}. The authors in \cite{Lagrange_sparse_coding} consider a \emph{sparsification} technique of the input data to leverage the sparsity properties of the distributed tasks. They also adapt their scheme to ensure {perfect} privacy against colluding {\workers}. %
The work in \cite{Matrix_sparsification} considers the case where a dense input matrix is sparsified to speed up the computation. %

\textit{Contributions and organization:} %
We are interested in preserving both, privacy and sparsity, of the input matrix. Insisting on perfect privacy does not allow the creation of sparse tasks. %
We open the door to sparse and private coded computing by relaxing the privacy requirement. 
We set the notation and formally explain the system model in Section~\ref{sec:Preliminaries}. In Section~\ref{sec:Schemes_trading_privacy_for_sparsity} we introduce the trade-off between sparsity and privacy through a coding strategy and optimize this trade-off for the setting at hand. We combine, in Section~\ref{sec:coded_comp_scheme}, our strategy with a task distribution scheme to create a coded computing scheme that tolerates stragglers and trades sparsity for privacy.%

%% file: preliminaries.tex
\emph{Notation:} We denote matrices and vectors by uppercase and lowercase bold letters, e.g., $\bfX$ and $\bfx$, respectively. The $(i,j)$-th entry of a matrix $\bfX$ is denoted by $\bfXij$. Random variables are denoted by uppercase \emph{typewriter} letters, e.g., $\rvY$. The random variables representing a matrix $\bfX$ and its $(i,j)$-th entry $\bfXij$ are denoted respectively by $\rvX$ and $\Xij{i}{j}$. A finite field of cardinality $q$ is denoted by $\F_q$ and its multiplicative group is denoted by $\F_q^*$, i.e., $\F_q^* = \F_q\setminus \{0\}$. Sets are denoted by calligraphic letters, e.g., $\mathcal{X}$. For a positive integer $b$, the set $\{1,2,\dots, b\}$ is denoted by $[b]$. Given $b$ random variables $\rvY_1,\dots,\rvY_b$ and a set $\mathcal{I}\subseteq [b]$, the set $\{\rvY_i\}_{i\in\mathcal{I}}$ contains the random variables indexed by $\mathcal{I}$, i.e., $\{\rvY_i\}_{i\in\mathcal{I}} \triangleq \{\rvY_{i} | i\in \mathcal{I}\}$. The $q$-ary entropy of a random variable $\rvX \in \mathbb{F}_q$ is denoted by $\textrm{H}_q(\rvX)$ and $\textrm{H}_q(\left[p_1, p_2, \dots, p_q \right])$ interchangeably, where $\left[p_1, p_2, \dots, p_q \right]$ is the probability mass function (PMF) of $\rvX$ over $\mathbb{F}_q$, i.e., $\Pr(\rvX = i) = p_i$ for all $i\in \mathbb{F}_q$. The $q$-ary mutual information between two random variables $\rvX$ and $\rvY$ is denoted by $\I_q(\rvX; \rvY)$. 
The indicator function $\mathbbm{1}_{\text{condition}}$ is one if ``condition" is true, and zero otherwise.
We define the sparsity of a matrix as follows.

{\definition(The sparsity level of a matrix $\bfX$) {A matrix $\bfX$ whose entries are independently and identically distributed (i.i.d.) has a sparsity level $\sparsity(\bfX)$ equal to the probability of its $(i,j)$-th entry being equal to $0$, i.e., $$ \sparsity(\bfX) = \Pr\{\Xij{i}{j} = 0\}$$}}

\emph{System model:} We consider the scenario where a {\master} node owns a private large sparse matrix $\bfA$ and a public vector $\bfx$. The matrix $\bfA \in \F_q^{m \times n}$ has a sparsity level S$(\bfA) = s > q^{-1}$. We assume that the entries of the matrix $\bfA$ are i.i.d with
$\Pr\{\Aij = 0\} = s$ and $\Pr\{\Aij = a\} = \frac{1-s}{q-1}$ where $a \in \F_q^*$. The vector $\bfx \in \F_q^{n \times k}$  is assumed to be uniformly distributed over $\F_q$, i.e., S$(\bfx) = q^{-1}$. The {\master} is interested in computing $\bfy = \bfA \bfx$. %
The {\master} distributes the computation to external nodes (referred to as {\workers}) that are hired from two different \emph{non-communicating} clusters. The {\workers} of the clusters have the following properties: 
\begin{enumerate}
    \item \emph{Untrusted cluster:} This cluster consists of $N_1$ {\workers}, $w_i^u$, for $i = 1, \dots, N_1$, that are fully untrusted. No leakage of information about $\bfA$ to those {\workers} is tolerated, i.e., perfect information-theoretic privacy is required here. {The only guarantee that the {\master} has from this cluster is that the {\workers} are honest but curious.} %
    \item \emph{Partly trusted cluster:} This cluster consists of $N_2$ {\workers}, $w_i^t$, for $i = 1, \dots, N_2$. The {\workers} of this cluster are honest but curious. However, a known limit of the {\workers} {(up to $z$)} collude to eavesdrop on the data of the {\master}. In addition, leaking a small amount of information about the matrix $\bfA$ to those {\workers} is tolerated.
    \item \emph{Stragglers:} The {\workers} may be assigned several computational tasks. \emph{Full} stragglers are unresponsive {\workers}. In contrast, \emph{partial} stragglers return part of their computational tasks to the {\master}.
\end{enumerate}

Our privacy measure is information-theoretic privacy. Given the random variables $\rvA$ and $\rvB$, we say that observing a realization $\bfB$ of $\rvB$ leaks $\varepsilon \triangleq \mutinf(\rvA;\rvB)$ information about $\rvA$. If the leakage $\varepsilon$ is zero, we say that perfect privacy is attained.

%% file: sparse_scheme.tex
We provide {a} coding strategy %
that encodes a sparse matrix $\bfA$, with $s\triangleq \sparsity(\bfA) >q^{-1}$, into two matrices $\bfB_1$ and $\bfB_2$ such that: \begin{enumerate*}[label={\textit{\roman*)}}] \item $\bfB_1$ and $\bfB_2$ have a desired sparsity level $s_1$ and $s_2$, $q^{-1}<s_1,s_2\leq s$; \item $\bfB_1$ and $\bfB_2$ leak a limited amount of information $\varepsilon$ about $\bfA$, i.e., $\mutinf(\rvA;\rvB_1)+\mutinf(\rvA;\rvB_2) \leq \varepsilon$; and \item $\bfA$ can be decoded from $\bfB_1$ and $\bfB_2$. %
\end{enumerate*}
{Our strategy provides a trade-off between the sparsity levels of the encoded tasks $s_1,s_2$ and the overall leakage $\varepsilon$. We base our scheme on Shannon's one-time pad~\cite{shannon_one_time_pad}, where the {\master} generates a random matrix $\bfR$ that is as big as $\bfA$. The encoded matrices are $\bfB_1 = \bfR$ and $\bfB_2 = \bfA+\bfR$, where the addition is entry-wise. We restrict our attention to the setting where $\mutinf(\rvA;\rvB_2) = 0$, and generalize our results in a future work.}

\emph{Dependent sparse one-time pad:} If $\bfR$ is generated uniformly at random from $\F_q$, then perfect privacy is achieved, but no sparsity in $\bfR$ and $\bfA+\bfR$ is maintained. To overcome the issue of obtaining dense matrices, we allow the random matrix $\bfR$ to be generated \emph{dependently} from the input private matrix $\bfA$. The main idea is to design a probability distribution to generate $\bfR$ such that good sparsity levels and low leakage are guaranteed. %
Since the entries $\rvA_{i,j}$ of $\rvA$ are assumed to be independent, we treat the entries of $\rvR$ independently. We define two conditional PMFs for the generation of $\bfRij$ given as follows.%
\begin{align}
 \label{eq:dependent_on_0}
    \Pr\{\Rij = r \lvert \Aij = 0\} &= \begin{cases} 
    \pz, &r = 0 \\   \dfrac{1-\pz}{q-1}, &r \neq 0,
    \end{cases}\\
\label{eq:dependent_on_nz}
    \Pr\{\Rij = r \lvert \Aij = a\} &= \begin{cases} 
    \pnz, &r = -a \\   \dfrac{1-\pnz}{q-1}, &r \neq -a.
    \end{cases}
\end{align}
where $r\in \F_q$, $a \in \F_q^*$, $-a$ is the additive inverse of $a$ in $\F_q^*$ and $\pz$ and $\pnz$ are non-negative numbers smaller than $1$. This strategy allows the \emph{padded matrix} $\bfA + \bfR$ to inherit some zero entries from $\bfA$ unless $\pz = 0$, see~\eqref{eq:dependent_on_0}. Moreover, the padded matrix will have some zero entries in the positions where the matrix $\bfA$ has non-zero entries for the case where $\pnz \neq 0$, see~\eqref{eq:dependent_on_nz}. The results are formally stated next.
\begin{lemma}\label{lemma:sparsity}
Given an input matrix $\bfA$ with a sparsity level $\sparsity(\bfA) = s$ and a random matrix $\bfR$ generated as shown in~\eqref{eq:dependent_on_0} and~\eqref{eq:dependent_on_nz}, the sparsity level of the padded matrix and the random matrix are given by $\sparsity\left(\bfA + \bfR\right)=(\pz - \pnz)s + \pnz$ and $\sparsity(\bfR) = \pz s + (1-\pnz)\frac{(1-s)}{q-1}$. The leakage about the input matrix is quantified by $\mutinf(\rvA+\rvR;\rvA) = m n \leakagep{1}$ and $\mutinf(\rvR;\rvA) = m n  \leakagep{2}$, where $\leakagep{1}$ and $\leakagep{2}$ are given in~\eqref{eq:final_res} and~\eqref{eq:leakage_from_R}, respectively.
\end{lemma}

\begin{remark}
Analyzing the dependency of the total leakage $m n (\leakagep{1}+\leakagep{2})$ and the total sparsity $\sparsity(\bfR)+\sparsity(\bfA+\bfR)$ on $p_{z0}$ and $p_{nz0}$ allows us to understand the nature of the tradeoff between sparsity and privacy. Preliminary observations show that both those quantities are increasing in $p_{z0}$ and $p_{nz0}$ for certain regimes. In this work we study the trade-off for the case where $p_{z0} = p_{nz0} = p$ (Lemma~\ref{lemma:R}) since it fits the model requirements, cf., Observation~\ref{obs:private}. We leave the general analysis for future investigation.
\end{remark}

\begin{proof}[Proof of Lemma~\ref{lemma:sparsity}]
We first compute $\sparsity(\bfA+\bfR)$. Note that the entries of the padded matrix $\bfA + \bfR$ are independently and identically distributed (i.i.d). That holds because the entries of $\bfA$ are assumed to be i.i.d. and the entry $\bfR_{i,j}$ of $\bfR$ only depends on $\bfA_{i,j}$. This allows us to write the following
\begin{align*}
\sparsity\left(\bfA + \bfR\right) &= \Pr\{ \Aij + \Rij = 0 \} \\
                                  &= \sum_{\ell = 0}^{q-1} \Pr\{ \Aij = \ell, \Rij = -\ell \}   \\
                                  &= \sum_{\ell = 0}^{q-1} \Pr\{ \Rij = -\ell\lvert \Aij = \ell \}\Pr\{\Aij = \ell\} \\
                                  &= (\pz - \pnz)s + \pnz.
\end{align*}
To compute $\sparsity(\bfR)$, we use the total law of probability to write
\begin{align}
    \sparsity\left(\bfR\right) &= \Pr\{\Rij = 0\} \nonumber \\
      &= \sum_{\ell=0}^{q-1} \Pr\{\Rij = 0\lvert \Aij =\ell\}\Pr\{\Aij = \ell\} \nonumber \\
                               & = \pz s + (1-\pnz)\frac{(1-s)}{q-1}. \label{eq:sparsity_R}
          \vspace{-0.2cm}
\end{align}%
\vspace{-0.2cm}

We now quantify the leakage of our coding scheme. To that end we quantify the leakage $\mutinf\left(\rvA+\rvR; \rvA\right)$ and $\mutinf\left(\rvR; \rvA\right)$, respectively.
Since the entries of $\bfA +\bfR$ and $\bfR$ are i.i.d., then 
\begin{align*}
    \mutinf\left(\rvA+\rvR;\rvA\right) &= m n  \mutinf \left(\Aij + \Rij; \Aij \right) \\
                                    &\triangleq m n \leakagep{1},\\
    \mutinf\left(\rvR; \rvA\right) &= m n\mutinf\left(\Rij; \Aij\right) \\
    &\triangleq m n\leakagep{2},
\end{align*}%
for any $i\in[m]$ and $j\in[n]$. Let $\leakagep{1}$ denote the element-wise leakage from the padded matrix $\bfA + \bfR$, given by
\vspace{-.3cm}
\begin{align}
    &\leakagep{1} \nonumber \\
    &= \entropy\left(\Aij+\Rij\right) - \entropy\left(\Aij+\Rij \lvert \Aij\right) \nonumber\\
             &= \entropy\left(\Aij+\Rij\right) - \entropy\left(\Rij \lvert \Aij\right) \label{eq:conditioned} \\ 
             &= \entropy\left(\Aij+\Rij\right) \nonumber \\ 
             &\quad - \sum_{\ell = 0}^{q-1} \entropy\left(\Rij \lvert \Aij = \ell\right)\Pr\{\Aij = \ell\} \label{eq:def_cond}\\
             &= \entropy\left(\left[\sparsity(\bfA+\bfR), \frac{1 - \sparsity(\bfA+\bfR)}{q-1},\dots \right]\right) \nonumber\\ %
             &\quad-s\entropy\left(\left[\pz,\frac{1-\pz}{q-1},\dots\right]\right) \nonumber  \\  %
             &\quad-(1-s)\entropy\left(\left[\pnz,\frac{1-\pnz}{q-1},\dots\right]\right), \label{eq:final_res} %
\end{align}
where \eqref{eq:conditioned} and~\eqref{eq:def_cond} hold due to the properties of conditional entropy and \eqref{eq:final_res} is obtained by writing the entropies in terms of the PMFs of their respective random variables. Following similar steps, the element-wise leakage from the \emph{padding matrix} $\bfR$ is given by
\begin{align}
    \leakagep{2} %
    &=\entropy\left(\left[\sparsity(\bfR), \dfrac{1-\sparsity(\bfR)}{q-1},\dots\right]\right) \nonumber \\
    &\quad-s\entropy\left(\left[\pz,\frac{1-\pz}{q-1},\dots\right]\right) \nonumber  \\  %
    &\quad-(1-s)\entropy\left(\left[\pnz,\frac{1-\pnz}{q-1},\dots\right]\right) \label{eq:leakage_from_R}.
\end{align}
\end{proof}
The derivation of~\eqref{eq:final_res} leads us to the following crucial observation on which we rely to build our coding scheme.

\begin{observation}\label{obs:private}
{For the special case of $\pz = \pnz = p$, the matrix $\bfA+\bfR$ leaks no information about $\bfA$, i.e., $\mutinf({\rvA+\rvR;\rvA})= mn  \leakage{1} = 0$.}
\end{observation}

The sparsity level of the padded matrix, $\sparsity(\bfA+\bfR) = p$ in this case, is pleasantly controllable with $p$ and does not impose any trade-off with privacy. There is however a trade-off between the sparsity level of $\bfR$ and its leakage about $\bfA$ as shown in Lemma~\ref{lemma:R}.%
\begin{figure*}[t]
    \centering
    \resizebox{.95\textwidth}{!}{\input{toy_example_rpm3}}
    \caption{A depiction of our coded computing scheme. The two non-communicating clusters, namely the untrusted and partly trusted one are illustrated on the left and right hand side, respectively. The {\workers} of the untrusted cluster $w_1^u, w_2^u, \dots, w_{N_1}^u$ get each $\alpha'$ tasks (created as in Section~\ref{subsec:colluding}) that leak nothing about the input matrix. On the other hand, each {\worker} of the partly trusted cluster $w_1^t, w_2^t, \dots, w_{N_2}^t$ gets $\alpha$ tasks that leak some information about the input matrix. Every {\worker} has to multiply the designated matrices with the public vector $\bfx$ and then sent each computation back to the {\master}.}
    \label{fig:fractional_repetition_sparse}
\end{figure*}
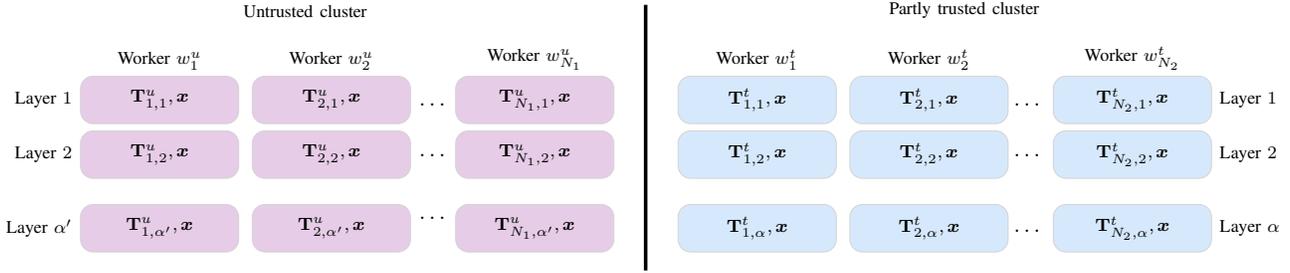

\begin{lemma}\label{lemma:R}
For the case when $\pz=\pnz=p$, if the padding matrix $\bfR$ is generated as in~\eqref{eq:dependent_on_0} and~\eqref{eq:dependent_on_nz}, then the increase of $p$ yields a higher sparsity level of $\bfR$, but it also increases its leakage about $\bfA$. 
\end{lemma}

\begin{proof}
In this case, from~\eqref{eq:sparsity_R}, the sparsity level of the padding matrix is given by
\begin{align*}
    \sparsity\left(\bfR\right)  &= p\frac{(sq-1)}{q-1} + \frac{(1-s)}{q-1}.
\end{align*}
Hence, for a fixed sparsity level of $\bfA$, $s>q^{-1}$, and a fixed field size $q$, the sparsity level of $\bfR$, $\sparsity(\bfR)$, is a linear monotonously increasing function of $p$. %

We now show that the leakage $\leakage{p} = \mutinf\left(\Rij;\Aij\right)$ from the padding matrix is an increasing function of $p$ when $p>q^{-1}$. From~\cite[Theorem 2.7.4]{cover2012elements}, the leakage is a convex function {of $p$} since the PMF of $\Aij$ is fixed and the conditional PMF of $\Rij\lvert \Aij$ {can be written as a convex mixture of two conditional distributions as follows}.
\begin{align*}
    \Pr (\Rij = r|\Aij = a) = p \mathbbm{1}_{r = -a} + (1-p)\frac{1}{1-q} \mathbbm{1}_{r \neq -a}.
\end{align*}

From perfect privacy, we know that the leakage is zero when $p= q^{-1}$. Since the mutual information is positive, then $p=q^{-1}$ is the global minimum of $\mutinf\left(\Rij;\Aij\right)$. Therefore, for $p > q^{-1}$ the leakage function is a monotonously increasing function. Thus, the increase of $p$ increases the sparsity level of $\bfR$ but also increases its leakage about $\bfA$.
\end{proof} 
As a result of Observation~\ref{obs:private} and Lemma~\ref{lemma:R}, when using this scheme for $p=\pz=\pnz$, the best choice $p^*$ is the maximum value of $p$ for which {$\mutinf\left(\rvR;\rvA\right)\leq\varepsilon$, since $\mutinf\left(\rvA + \rvR;\rvA\right) = 0$.}

%% file: toy_example_rpm3.tex
	\begin{tikzpicture}[>=stealth', auto,
		triangle/.style = {fill=white, regular polygon, regular polygon sides=3 }]
		\definecolor{blue}{rgb}{0.19, 0.55, 0.91}
		\definecolor{plum}{rgb}{0.8, 0.6, 0.8}
		\def\blue{\color{blue}}
		\def\ogreen{\color{ogreen}}
		\def\orange{\color{orange}}
		\def\nred{\color{plum}}
		\definecolor{lightgray}{rgb}{0.83, 0.83, 0.83}
		\def\mx{0.4}
		\def\my{1}
		\tikzset{auto shift/.style={auto=right,->,
				to path={ let \p1=(\tikztostart),\p2=(\tikztotarget),
					\n1={atan2(\y2-\y1,\x2-\x1)},\n2={\n1+180}
					in ($(\tikztostart.{\n1})!1mm!270:(\tikztotarget.{\n2})$) -- 
					($(\tikztotarget.{\n2})!1mm!90:(\tikztostart.{\n1})$) \tikztonodes}}}
		\tikzstyle{matrixa} = [rectangle, rounded corners = 2mm, draw = lightgray, minimum height = 2.2cm, minimum width = 1.5 cm, fill = blue!40]
		\tikzstyle{matrixr} = [rectangle, rounded corners = 2mm, draw = lightgray, minimum height = 2.2cm, minimum width = 1.5 cm, fill = blue!20]
		
		\tikzstyle{matrixa1} = [rectangle, rounded corners = 1mm, draw = lightgray, minimum height = 0.75 cm, minimum width = 1.5 cm, fill = blue!40]
		\tikzstyle{matrixr1} = [rectangle, rounded corners = 2mm, draw = lightgray, minimum height = 0.75cm, minimum width = 1.5 cm, fill = blue!20]
				
		\tikzstyle{matrixa2} = [rectangle, rounded corners = 1mm, draw = lightgray, minimum height = 0.75cm, minimum width = 1.5 cm, fill = blue!40]
		
		\tikzstyle{matrixb} = [rounded corners = 2mm, draw = lightgray, minimum height = 1.5cm, minimum width = 0.3 cm, fill = ogreen!30]
		
		\tikzstyle{tasks1} = [rounded corners = 2mm, draw = lightgray, minimum height = 0.75cm, minimum width = 2.5 cm, fill = plum!50]
		
		\tikzstyle{tasks2} = [rounded corners = 2mm, draw = lightgray, minimum height = 0.75cm, minimum width = 2.5 cm, fill = blue!20]

		\tikzstyle{server} = [fill=black!10, rectangle, rounded corners=4mm, draw,minimum width=2em, minimum height=2.5em]
		\node[inner sep=0] (s1) at (-1,0) %
		{};

		\node[inner sep=0pt, below left = 0.6 and 6 of s1, font=\footnotesize] (w1){};%
		\node[below =-0.3 of w1] (ww'2){};%
		\node[inner sep=0pt, left = 2 of w1, font=\footnotesize] (w'1){};%
		\node[below =-0.3 of w'1] (ww'1){};%
		\node[inner sep=0pt, right = 3cm of w1, font=\footnotesize] (w2){};%
		\node[below =-0.3 of w2] (ww'N){};%
		\node[inner sep=0pt, right = 3cm of w2, font=\footnotesize] (w3){};%
		\node[below =-0.3 of w3] (ww1){};%
		\node[inner sep=0pt, right = 2cm of w3, font=\footnotesize] (w4){};%
		\node[below =-0.3 of w4] (ww2){};%
		\node[inner sep=0pt, right = 3cm of w4, font=\footnotesize] (w5){};%
		\node[below =-0.3 of w5] (wwN){};%
		\node[tasks1, below left = 0 and 0 of ww'1] (A1+R1) {\footnotesize$\bfT_{1, 1}^u, \bfx$};
		\node[tasks1, right = 0.2 of A1+R1] (A2+R2) {\footnotesize$\bfT_{2,1}^u, \bfx$};
		\node[tasks1, right = 0.7 of A2+R2] (AN+RN) {\footnotesize$\bfT_{N_1, 1}^u, \bfx$};
		\node[tasks1, below = 0.1 of A1+R1] (A'1+R'1) {\footnotesize$\bfT_{1,2}^u, \bfx$};
		\node[tasks1, below = 0.1 of A2+R2] (A'2+R'2) {\footnotesize$\bfT_{2,2}^u, \bfx$};
		\node[tasks1, below = 0.1 of AN+RN] (A'N+R'N) {\footnotesize$\bfT_{N_1,2}^u, \bfx$};
		\node[tasks1, below = 0.4 of A'1+R'1] (Al+Rl) {\footnotesize$\bfT^u_{1, \alpha'}, \bfx$};
		\node[tasks1, below = 0.4 of A'2+R'2] (A'l+R'l) {\footnotesize$\bfT^u_{2, \alpha'}, \bfx$};
		\node[tasks1, below = 0.4 of A'N+R'N] (A''l+R''l) {\footnotesize$\bfT^u_{N_1, \alpha'}, \bfx$};
		\node[above = 2 of Al+Rl] (W'1) {\footnotesize {\Worker} $w^u_1$};
		\node[above = 2 of A'l+R'l] (W'2) {\footnotesize {\Worker} $w^u_2$};
		\node[above = 2 of A''l+R''l] (W'N) {\footnotesize {\Worker} $w^u_{N_1}$};
		\node[tasks2, right = 1 of AN+RN] (R1) {\footnotesize$\bfT^t_{1, 1}, \bfx$};
		\node[tasks2, right = 0.2 of R1] (R2) {\footnotesize$\bfT^t_{2, 1}, \bfx$};
		\node[tasks2, right = 0.7 of R2] (RN) {\footnotesize$\bfT^t_{N_2, 1}, \bfx$};
		\node[tasks2, below = 0.1 of R1] (R'1) {\footnotesize$\bfT^t_{1, 2}, \bfx$};
		\node[tasks2, below = 0.1 of R2] (R'2) {\footnotesize$\bfT^t_{2, 2}, \bfx$};
		\node[tasks2, below = 0.1 of RN] (R'N) {\footnotesize$\bfT^t_{N_2, 2}, \bfx$};
		\node[tasks2, below = 0.4 of R'1] (R''1) {\footnotesize$\bfT^t_{1, \alpha}, \bfx$};				
		\node[tasks2, below = 0.4 of R'2] (R''2) {\footnotesize$\bfT^t_{2, \alpha}, \bfx$};
		\node[tasks2, below = 0.4 of R'N] (R''N) {\footnotesize$\bfT^t_{N_2, \alpha}, \bfx$};
		\node[above = 2 of R''1] (W1) {\footnotesize {\Worker} $w^t_1$};
		\node[above = 2 of R''2] (W2) {\footnotesize {\Worker} $w^t_2$};
		\node[above = 2 of R''N] (WN) {\footnotesize {\Worker} $w^t_{N_2}$};
		
		\node[thick,above right = 0.8 and 0.75 of w1](pikat1){};%
		\node[thick,above right = 0.8 and 2 of w4](pikat'1){};%
		\node[thick,below = 1.1 of pikat1](pikat2) {$\dots$};
		\node[thick,below = 0.5 of pikat2](pikat3) {$\dots$};
		\node[thick,below = 0.7 of pikat3](pikat4) {$\dots$};
		\node[thick,below = 1.1 of pikat'1](pikat'2) {$\dots$};
		\node[thick,below = 0.5 of pikat'2](pikat'3) {$\dots$};
		\node[thick,below = 0.9 of pikat'3](pikat'4) {$\dots$};
		\draw[ultra thick] (-2.8,0.5) -- (-2.8, -3.7);
		\node[above left = 0.8 and 0 of w1](cluster1) {\footnotesize Untrusted cluster};
		\node[above right = 0.8 and -0.2 of w4](cluster2) {\footnotesize Partly trusted cluster};
		
		\node[left = 0 of A1+R1]{\footnotesize Layer $1$};
		\node[left = 0 of A'1+R'1]{\footnotesize Layer $2$};
		\node[left = 0 of Al+Rl]{\footnotesize Layer $\alpha'$};
		
		\node[right = 0 of R''N]{\footnotesize Layer $\alpha$};
        \node[right = 0 of R'N]{\footnotesize Layer $2$};
        \node[right = 0 of RN]{\footnotesize Layer $1$};
        
		\end{tikzpicture}

%% file: fractional_repetition_scheme.tex
In this section we combine our coding strategy with the cyclically shifted task assignment from~\cite{coded_sparse_matrix_leverages_partial_stragglers} and~\cite{el2010fractional} to obtain our coded computing scheme that is resilient to stragglers. The reason for choosing this {assignment} scheme is that it perfectly preserves the {sparsity} of the input tasks, $\bfR$ and $\bfA+\bfR$.%

\subsection{Task creation and distribution}\label{subsec:colluding}

The {\master} observes the input matrix $\bfA \in \F_q^{m \times n}$ and creates a matrix $\bfR \in \F_q^{m \times n}$ as described in~\eqref{eq:dependent_on_0} and~\eqref{eq:dependent_on_nz} for the case $\pz = \pnz =p$. Then, the {\master} splits the padded matrix $\bfA + \bfR$ row-wise into $N_1$ sub-matrices $\left[(\bfA+\bfR)_1^\transpose, (\bfA+\bfR)_2^\transpose,\dots,(\bfA+\bfR)_{N_1}^\transpose\right]^\transpose$, where $(\bfA+\bfR)_i \in \F_q^{\frac{m}{N_1}\times n}$ for $i \in [N_1]$. Similarly, the {\master} splits $\bfR$ into $N_2$ submatrices $\bfR = \left[\bfR_1^\transpose, \bfR_2^\transpose,\dots,\bfR_{N_2}^\transpose\right]^\transpose$, where $\bfR_k \in \F_q^{\frac{m}{N_2}\times n}$ for $k \in [N_2]$. Then, the {\master} %
creates $N_1$ tasks $\bfT_{i, 1}^u = (\bfA + \bfR)_i$ each assigned to {\worker} $w_i^u$ for $i\in [N_1]$. The collection of those tasks is referred to as the first \emph{layer} of tasks.
A similar first layer of tasks $\bfT_{k, 1}^t = \bfR_k$ for $k \in [N_2]$ is created and assigned to {\workers} of  the partly trusted cluster, cf., Fig.~\ref{fig:fractional_repetition_sparse}. %
The {\master} creates %
the tasks of the other layers %
as 
\begin{align*}
\bfT^u_{(i \mod N_1)+1,j} &= \bfT^u_{i,j-1},\quad i \in [N_1], j \in [\alpha'] \setminus \{1\},\\
\bfT^t_{(k \mod N_2)+1, b} &= \bfT^t_{k, b-1},\quad k \in [N_2], b \in [\alpha] \setminus \{1\}.
\end{align*}
The {\master} publishes $\bfx$ to all the {\workers} and assigns to {\worker} $w_i^u$, $i\in[N_1]$, the tasks $\{\bfT^u_{i,j}\}_{j=1}^{\alpha'}$ and to {\worker} $w_k^t$, $k\in[N_2]$, the tasks $\{\bfT^u_{k,b}\}_{b=1}^{\alpha}$. Each {\worker} multiplies its assigned tasks by $\bfx$ and sends the results back to the {\master}.

\subsection{Analysis of the scheme}\label{subsec:privacy_scheme}
\begin{theorem}\label{thm:main}
Using the task creation strategy explained above, the {\workers} of the untrusted cluster learn nothing about the input matrix $\bfA$, i.e., $\mutinf(\{\rvT_{i,j}^u\}_{i\in[N_1],j\in[\alpha']};\rvA) = 0.$
The leakage to the colluding {\workers} of the partly trusted cluster about the input matrix $\bfA$ is quantified by 
{\begin{equation}\label{eq:colluding_trusted}
    \mutinf(\{\rvT^t_{i,j}\}_{i\in\cZ, j\in [\alpha]};\rvA) = \min\left\{\dfrac{\alpha z}{N_2}, 1\right\}\cdot m \cdot n \cdot \leakage{2}.
\end{equation}}
As a result, the maximum sparsity allowed for the assigned tasks is given by
\begin{align*}
    \sparsity(\bfA+\bfR) &= p^\star, \quad \text{ and } \quad\sparsity(\bfR) = p^* \frac{sq-1}{q-1}+\frac{1-s}{q-1},
\end{align*}
where $p^* = \min_{\mutinf(\{\rvT^t_{i,j}\}_{i\in\cZ, j\in [\alpha]};\rvA)\leq \varepsilon} p$.
The {\master} can obtain the desired computation after receiving any $K^u \triangleq \frac{-\alpha'^2 +\alpha'(2N_1-1)}{2} +1$ and any $K^t \triangleq \frac{-\alpha^2 +\alpha(2N_2-1)}{2} +1$ responses from the {\workers} of the untrusted and partly trusted clusters, respectively. Thus allowing a tolerance of partial stragglers. In case of full stragglers, the {\master} can tolerate up to $\alpha - 1$ and $\alpha' - 1$ stragglers from the respective clusters.
\end{theorem}

\begin{proof}%
\emph{Privacy:} We start by proving that $\mutinf(\{\rvT_{i,j}^u\}_{i\in[N_1],j\in[\alpha']};\rvA) = 0$. Recall that for the first layer of tasks, $\bfT_{i,1}^u = (\bfA + \bfR)_i$, $i\in [N_1]$. Those tasks are independent from each other, cf., the proof of Lemma~\ref{lemma:sparsity}. This holds because by construction, the entries of the padding matrix $\bfR$ are generated independently from each other. Generating entry $\bfRij$ depends only on $\bfAij$ which is also assumed to be independent from the other entries of $\bfA$. Hence, 
\begin{align*}
    \mutinf(\rvT_{1,1}^u,\dots,\rvT^u_{N_1,1};\rvA) & = \sum_{i=1}^{N_1}\mutinf(\rvT^u_{i,1};\rvA)
    = 0.%
\end{align*}
The last equality holds from Observation~\ref{obs:private}. The tasks allocated in the other layers $\{\bfT^u_{i, j}\}_{i \in [N_1], 2\leq j \leq \alpha'}$ do not add any information to the {\workers} about $\bfA$ since they are copies of $\bfT_{i,1}^u$, i.e., using the chain rule of mutual information
\begin{align*}
    \mutinf(\{\rvT_{i,j}^u\}_{i\in[N_1],j\in[\alpha']};\rvA) & = \mutinf(\{\rvT^u_{i,1}\}_{i\in [N_1]};\rvA)\\
    & \hspace{-1.5cm} + \mutinf(\{\rvT_{i,j}^u\}_{i \in [N_1], 2\leq j \leq \alpha'};\rvA|\{\rvT^u_{i,1}\}_{i\in [N_1]})\\
    & =  0.
\end{align*}
We now prove~\eqref{eq:colluding_trusted}. By construction, the set of tasks sent to any $z$ {\workers} of the partly trusted clusters $\{\bfT^t_{i,j}\}_{i \in \cZ, j\in [\alpha]}, \cZ \subseteq [N_2] ,\lvert \cZ \rvert = z$  can include at most $\min\{\alpha z, N_2\}$ unique leaking tasks $\bfR_i, i \in [N_2]$. The set of tasks allocated in the first layer $\{\bfT^t_{i, 1}\}_{i \in [N_2]}$ consists of unique elements (due to the cyclic shifted assignment), hence we can write 
\begin{align}
    \mutinf(\{\rvT^t_{i,j}\}_{i\in \cZ, j\in[\alpha]}; \rvA) &= \mutinf(\rvT^t_{1,1}, \rvT^t_{2,1},\dots, \rvT^t_{\min{\{\alpha z, N_2\}},1}; \rvA) \nonumber \\
    &=\sum_{i=1}^{\min{\{\alpha z, N_2\}}}\mutinf(\rvT^t_{i,1};\rvA)\label{eq:proof_theorem_1} \\
    &= \min{\left\{\dfrac{\alpha z}{N_2}, 1\right\}} m n  \leakage{2}, \label{eq:proof_theorem_2}
\end{align}
where \eqref{eq:proof_theorem_1} holds since the tasks $\bfT_{i,1}^t$ are independent from each other and~\eqref{eq:proof_theorem_2} holds since the entries of the padding matrices $\bfRij$ are independent from each other. The leakage $\leakage{2}$ is quantified in~\eqref{eq:leakage_from_R} for the case $\pz=\pnz$.

\emph{Sparsity:} Given a desired privacy level $\varepsilon$, the only constraint this scheme has is $\mutinf(\{\rvT^t_{i,j}\}_{i\in\cZ, j\in [\alpha]};\rvA) \leq \varepsilon$. Since the sparsity levels of both $\bfR$ and $\bfA+\bfR$ are increasing in $p$, then the scheme gives the best sparsity guarantee by maximizing $p^*$ subject to the desired privacy constraint.

To better show the trade-off between sparsity and privacy, we define the relative leakage $\bar{\varepsilon} = \frac{\varepsilon}{\entropy(\rvA)} = \frac{\varepsilon}{mn\entropy(\Aij)}$. Hence, we have 
\begin{equation}\label{eq:individual_leakage}
    \leakage{2} \leq \dfrac{\bar{\varepsilon}\entropy(\Aij)}{\min\{\frac{\alpha z}{N_2},1\}}.
\end{equation}
We plot in Fig.~\ref{fig:my_label} the allowed sparsity $p^*$ as a function of $z$ and the privacy constraint. Observe that $p^*$ decreases with $z$ and the desired privacy and so does the obtained sparsity.

\begin{figure}[t]
    \centering
    \resizebox{.48\textwidth}{!}{\input{plot_for_paper}}
    \caption{Trade-off sparsity versus privacy. We take $N_2 = 100$, $\alpha = 1$, $s =\sparsity(\bfA)=0.93$ and consider the finite field $\F_{256}$. The value of $p^*$ is the maximum value of $p$ for which~\eqref{eq:individual_leakage} holds. The value of $\bar{\varepsilon}$ reflects the privacy requirement (see the proof of Theorem~\ref{thm:main}). The value $\bar{\varepsilon} = 0$ corresponds to perfect privacy and requires $p^* = q^{-1}$ for all $ z \in [N_2]$, whereas $\bar{\varepsilon} = 1$ allows our scheme to leak as much as a non-private scheme, hence allowing the choice $p^* = 1$ for all $ z \in [N_2]$. Recall that $\sparsity(\bfA + \bfR) = p^*$ and for this choice of $s$ and $q$ we have $\sparsity(\bfR) \approx 0.93p^*$. Hence, for $0 < \bar{\epsilon} < 1$, the increase of $z$ and/or the {decrease of} $\bar{\varepsilon}$ reduces the achievable sparsity level of $\bfR$ and $\bfA+\bfR$. 
    }
    \label{fig:my_label}
\end{figure}
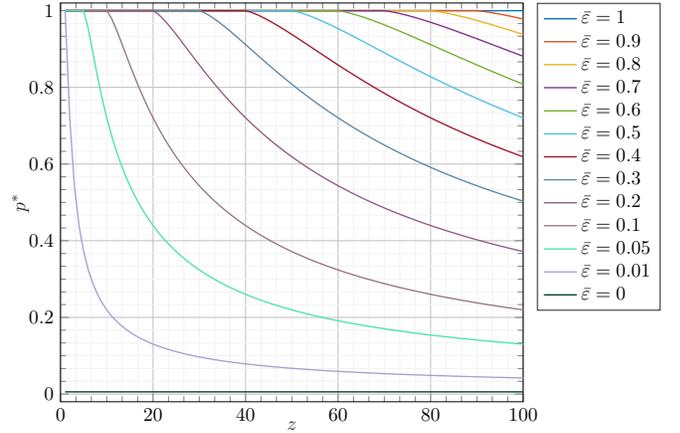

\emph{Straggler tolerance:} {The {\workers} are assumed to run the computations sequentially. Each {\worker} starts by computing the task of layer $1$ and proceeds to the next layer after completing this task. We use the same proof technique as in~\cite{coded_sparse_matrix_leverages_partial_stragglers}. Without loss of generality we will analyse the minimum number of responses needed from the partly trusted cluster. The worst case scenario is when only one task $\bfR_k\bfx$ for $k \in [N_2]$ is not retrieved from any of the {\workers}. The task $\bfR_k\bfx$ is only assigned to $\alpha$ {\workers}. Assume that all $N_2 - \alpha$ {\workers} that cannot compute $\bfR_k\bfx$ finished their computations. The remaining $\alpha$ {\workers} have $\bfR_k$ at layers $1,2,\dots, \alpha$, respectively. Therefore, the remaining $\alpha$ {\workers} can compute tasks different than $\bfR_k\bfx$ by performing at most $\sum_{u =1}^{\alpha} (u-1)$ redundant computations. Note that after having any $(N_2-\alpha)\alpha + \sum_{u =1}^{\alpha} (u-1) + 1$ computations the {\master} obtains $\bfR_k\bfx$ and therefore concludes the computation of $\bfR\bfx$. By simplyfing the equation we say that the {\master} needs to receive any $K^t \triangleq \frac{-\alpha^2 +\alpha(2N_2-1)}{2} +1$ responses from the {\workers} of the partly trusted cluster to obtain $\bfR\bfx$. Similarly, any  $K^u \triangleq \frac{-\alpha'^2 +\alpha'(2N_1-1)}{2} +1$ responses from the {\workers} of the untrusted cluster are enough to reconstruct $(\bfA+\bfR)\bfx$. Afterwards, the {\master} can obtain $\bfy = \bfA\bfx = (\bfA+\bfR)\bfx - \bfR\bfx$.} The statement about full stragglers holds by construction. %
\end{proof}

%% file: plot_for_paper.tex
\definecolor{mycolor1}{rgb}{0.00000,0.44700,0.74100}%
\definecolor{mycolor2}{rgb}{0.85000,0.32500,0.09800}%
\definecolor{mycolor3}{rgb}{0.92900,0.69400,0.12500}%
\definecolor{mycolor4}{rgb}{0.49400,0.18400,0.55600}%
\definecolor{mycolor5}{rgb}{0.46600,0.67400,0.18800}%
\definecolor{mycolor6}{rgb}{0.30100,0.74500,0.93300}%
\definecolor{mycolor7}{rgb}{0.63500,0.07800,0.18400}%
\definecolor{mycolor8}{rgb}{0.36, 0.54, 0.66}%
\definecolor{mycolor9}{rgb}{0.57, 0.36, 0.51}%
\definecolor{mycolor10}{rgb}{0.6, 0.47, 0.48}%
\definecolor{mycolor11}{rgb}{0.3, 0.88, 0.71}%
\definecolor{mycolor12}{rgb}{0.64, 0.64, 0.82}%
\definecolor{mycolor13}{rgb}{0.0, 0.26, 0.15}%
\begin{tikzpicture}

\begin{axis}[%
scale only axis,
xmin=0,
xmax=100,
xlabel style={font=\color{white!15!black}, yshift = 3mm},
xlabel={$z$},
ymin=-0.02,
ymax=1.02,
grid=both,
grid style={line width=.1pt, draw=gray!10},
major grid style={line width=.2pt,draw=gray!50},
ylabel style={font=\color{white!15!black}, yshift = -5mm},
minor tick num=5,
ylabel={$p^*$},
axis background/.style={fill=white},
legend style={legend cell align=left, align=left, draw=white!15!black},
legend pos = {outer north east},
every axis plot/.append style={line width =.7pt}
]
\addplot [color=mycolor1]
  table[row sep=crcr]{%
1	0.999996200203896\\
2	0.999996200203896\\
3	0.999996200203896\\
4	0.999996200203896\\
5	0.999996200203896\\
6	0.999996200203896\\
7	0.999996200203896\\
8	0.999996200203896\\
9	0.999996200203896\\
10	0.999996200203896\\
11	0.999996200203896\\
12	0.999996200203896\\
13	0.999996200203896\\
14	0.999996200203896\\
15	0.999996200203896\\
16	0.999996200203896\\
17	0.999996200203896\\
18	0.999996200203896\\
19	0.999996200203896\\
20	0.999996200203896\\
21	0.999996200203896\\
22	0.999996200203896\\
23	0.999996200203896\\
24	0.999996200203896\\
25	0.999996200203896\\
26	0.999996200203896\\
27	0.999996200203896\\
28	0.999996200203896\\
29	0.999996200203896\\
30	0.999996200203896\\
31	0.999996200203896\\
32	0.999996200203896\\
33	0.999996200203896\\
34	0.999996200203896\\
35	0.999996200203896\\
36	0.999996200203896\\
37	0.999996200203896\\
38	0.999996200203896\\
39	0.999996200203896\\
40	0.999996200203896\\
41	0.999996200203896\\
42	0.999996200203896\\
43	0.999996200203896\\
44	0.999996200203896\\
45	0.999996200203896\\
46	0.999996200203896\\
47	0.999996200203896\\
48	0.999996200203896\\
49	0.999996200203896\\
50	0.999996200203896\\
51	0.999996200203896\\
52	0.999996200203896\\
53	0.999996200203896\\
54	0.999996200203896\\
55	0.999996200203896\\
56	0.999996200203896\\
57	0.999996200203896\\
58	0.999996200203896\\
59	0.999996200203896\\
60	0.999996200203896\\
61	0.999996200203896\\
62	0.999996200203896\\
63	0.999996200203896\\
64	0.999996200203896\\
65	0.999996200203896\\
66	0.999996200203896\\
67	0.999996200203896\\
68	0.999996200203896\\
69	0.999996200203896\\
70	0.999996200203896\\
71	0.999996200203896\\
72	0.999996200203896\\
73	0.999996200203896\\
74	0.999996200203896\\
75	0.999996200203896\\
76	0.999996200203896\\
77	0.999996200203896\\
78	0.999996200203896\\
79	0.999996200203896\\
80	0.999996200203896\\
81	0.999996200203896\\
82	0.999996200203896\\
83	0.999996200203896\\
84	0.999996200203896\\
85	0.999996200203896\\
86	0.999996200203896\\
87	0.999996200203896\\
88	0.999996200203896\\
89	0.999996200203896\\
90	0.999996200203896\\
91	0.999996200203896\\
92	0.999996200203896\\
93	0.999996200203896\\
94	0.999996200203896\\
95	0.999996200203896\\
96	0.999996200203896\\
97	0.999996200203896\\
98	0.999996200203896\\
99	0.999996200203896\\
100	0.999996200203896\\
};
\addlegendentry{$\bar{\varepsilon} = 1$}

\addplot [color=mycolor2]
  table[row sep=crcr]{%
1	0.999996200203896\\
2	0.999996200203896\\
3	0.999996200203896\\
4	0.999996200203896\\
5	0.999996200203896\\
6	0.999996200203896\\
7	0.999996200203896\\
8	0.999996200203896\\
9	0.999996200203896\\
10	0.999996200203896\\
11	0.999996200203896\\
12	0.999996200203896\\
13	0.999996200203896\\
14	0.999996200203896\\
15	0.999996200203896\\
16	0.999996200203896\\
17	0.999996200203896\\
18	0.999996200203896\\
19	0.999996200203896\\
20	0.999996200203896\\
21	0.999996200203896\\
22	0.999996200203896\\
23	0.999996200203896\\
24	0.999996200203896\\
25	0.999996200203896\\
26	0.999996200203896\\
27	0.999996200203896\\
28	0.999996200203896\\
29	0.999996200203896\\
30	0.999996200203896\\
31	0.999996200203896\\
32	0.999996200203896\\
33	0.999996200203896\\
34	0.999996200203896\\
35	0.999996200203896\\
36	0.999996200203896\\
37	0.999996200203896\\
38	0.999996200203896\\
39	0.999996200203896\\
40	0.999996200203896\\
41	0.999996200203896\\
42	0.999996200203896\\
43	0.999996200203896\\
44	0.999996200203896\\
45	0.999996200203896\\
46	0.999996200203896\\
47	0.999996200203896\\
48	0.999996200203896\\
49	0.999996200203896\\
50	0.999996200203896\\
51	0.999996200203896\\
52	0.999996200203896\\
53	0.999996200203896\\
54	0.999996200203896\\
55	0.999996200203896\\
56	0.999996200203896\\
57	0.999996200203896\\
58	0.999996200203896\\
59	0.999996200203896\\
60	0.999996200203896\\
61	0.999996200203896\\
62	0.999996200203896\\
63	0.999996200203896\\
64	0.999996200203896\\
65	0.999996200203896\\
66	0.999996200203896\\
67	0.999996200203896\\
68	0.999996200203896\\
69	0.999996200203896\\
70	0.999996200203896\\
71	0.999996200203896\\
72	0.999996200203896\\
73	0.999996200203896\\
74	0.999996200203896\\
75	0.999996200203896\\
76	0.999996200203896\\
77	0.999996200203896\\
78	0.999996200203896\\
79	0.999996200203896\\
80	0.999996200203896\\
81	0.999996200203896\\
82	0.999996200203896\\
83	0.999996200203896\\
84	0.999996200203896\\
85	0.999996200203896\\
86	0.999996200203896\\
87	0.999996200203896\\
88	0.999996200203896\\
89	0.999996200203896\\
90	0.999996200203896\\
91	0.998753666877747\\
92	0.997112154960632\\
93	0.99525785446167\\
94	0.993190765380859\\
95	0.990971684455872\\
96	0.988631010055542\\
97	0.986138343811035\\
98	0.983554482460022\\
99	0.980849027633667\\
100	0.978082776069641\\
};
\addlegendentry{$\bar{\varepsilon} = 0.9$}

\addplot [color=mycolor3]
  table[row sep=crcr]{%
1	0.999996200203896\\
2	0.999996200203896\\
3	0.999996200203896\\
4	0.999996200203896\\
5	0.999996200203896\\
6	0.999996200203896\\
7	0.999996200203896\\
8	0.999996200203896\\
9	0.999996200203896\\
10	0.999996200203896\\
11	0.999996200203896\\
12	0.999996200203896\\
13	0.999996200203896\\
14	0.999996200203896\\
15	0.999996200203896\\
16	0.999996200203896\\
17	0.999996200203896\\
18	0.999996200203896\\
19	0.999996200203896\\
20	0.999996200203896\\
21	0.999996200203896\\
22	0.999996200203896\\
23	0.999996200203896\\
24	0.999996200203896\\
25	0.999996200203896\\
26	0.999996200203896\\
27	0.999996200203896\\
28	0.999996200203896\\
29	0.999996200203896\\
30	0.999996200203896\\
31	0.999996200203896\\
32	0.999996200203896\\
33	0.999996200203896\\
34	0.999996200203896\\
35	0.999996200203896\\
36	0.999996200203896\\
37	0.999996200203896\\
38	0.999996200203896\\
39	0.999996200203896\\
40	0.999996200203896\\
41	0.999996200203896\\
42	0.999996200203896\\
43	0.999996200203896\\
44	0.999996200203896\\
45	0.999996200203896\\
46	0.999996200203896\\
47	0.999996200203896\\
48	0.999996200203896\\
49	0.999996200203896\\
50	0.999996200203896\\
51	0.999996200203896\\
52	0.999996200203896\\
53	0.999996200203896\\
54	0.999996200203896\\
55	0.999996200203896\\
56	0.999996200203896\\
57	0.999996200203896\\
58	0.999996200203896\\
59	0.999996200203896\\
60	0.999996200203896\\
61	0.999996200203896\\
62	0.999996200203896\\
63	0.999996200203896\\
64	0.999996200203896\\
65	0.999996200203896\\
66	0.999996200203896\\
67	0.999996200203896\\
68	0.999996200203896\\
69	0.999996200203896\\
70	0.999996200203896\\
71	0.999996200203896\\
72	0.999996200203896\\
73	0.999996200203896\\
74	0.999996200203896\\
75	0.999996200203896\\
76	0.999996200203896\\
77	0.999996200203896\\
78	0.999996200203896\\
79	0.999996200203896\\
80	0.999996200203896\\
81	0.998571276664734\\
82	0.996671378612518\\
83	0.994497895240784\\
84	0.992096424102783\\
85	0.98951256275177\\
86	0.98677670955658\\
87	0.983888864517212\\
88	0.980849027633667\\
89	0.977717995643616\\
90	0.974465370178223\\
91	0.971121549606323\\
92	0.967716932296753\\
93	0.964190721511841\\
94	0.960603713989258\\
95	0.956955909729004\\
96	0.953247308731079\\
97	0.949508309364319\\
98	0.945708513259888\\
99	0.941878318786621\\
100	0.937987327575684\\
};
\addlegendentry{$\bar{\varepsilon} = 0.8$}

\addplot [color=mycolor4]
  table[row sep=crcr]{%
1	0.999996200203896\\
2	0.999996200203896\\
3	0.999996200203896\\
4	0.999996200203896\\
5	0.999996200203896\\
6	0.999996200203896\\
7	0.999996200203896\\
8	0.999996200203896\\
9	0.999996200203896\\
10	0.999996200203896\\
11	0.999996200203896\\
12	0.999996200203896\\
13	0.999996200203896\\
14	0.999996200203896\\
15	0.999996200203896\\
16	0.999996200203896\\
17	0.999996200203896\\
18	0.999996200203896\\
19	0.999996200203896\\
20	0.999996200203896\\
21	0.999996200203896\\
22	0.999996200203896\\
23	0.999996200203896\\
24	0.999996200203896\\
25	0.999996200203896\\
26	0.999996200203896\\
27	0.999996200203896\\
28	0.999996200203896\\
29	0.999996200203896\\
30	0.999996200203896\\
31	0.999996200203896\\
32	0.999996200203896\\
33	0.999996200203896\\
34	0.999996200203896\\
35	0.999996200203896\\
36	0.999996200203896\\
37	0.999996200203896\\
38	0.999996200203896\\
39	0.999996200203896\\
40	0.999996200203896\\
41	0.999996200203896\\
42	0.999996200203896\\
43	0.999996200203896\\
44	0.999996200203896\\
45	0.999996200203896\\
46	0.999996200203896\\
47	0.999996200203896\\
48	0.999996200203896\\
49	0.999996200203896\\
50	0.999996200203896\\
51	0.999996200203896\\
52	0.999996200203896\\
53	0.999996200203896\\
54	0.999996200203896\\
55	0.999996200203896\\
56	0.999996200203896\\
57	0.999996200203896\\
58	0.999996200203896\\
59	0.999996200203896\\
60	0.999996200203896\\
61	0.999996200203896\\
62	0.999996200203896\\
63	0.999996200203896\\
64	0.999996200203896\\
65	0.999996200203896\\
66	0.999996200203896\\
67	0.999996200203896\\
68	0.999996200203896\\
69	0.999996200203896\\
70	0.999996200203896\\
71	0.998312890529633\\
72	0.996078610420227\\
73	0.993494749069214\\
74	0.990652501583099\\
75	0.987567067146301\\
76	0.984314441680908\\
77	0.980849027633667\\
78	0.977262020111084\\
79	0.973523020744324\\
80	0.969662427902222\\
81	0.965710639953613\\
82	0.961637258529663\\
83	0.957503080368042\\
84	0.953247308731079\\
85	0.948991537094116\\
86	0.944614171981812\\
87	0.940236806869507\\
88	0.935768246650696\\
89	0.931299686431885\\
90	0.926800727844238\\
91	0.922240972518921\\
92	0.917681217193604\\
93	0.913121461868286\\
94	0.908531308174133\\
95	0.90394115447998\\
96	0.899381399154663\\
97	0.894760847091675\\
98	0.890201091766357\\
99	0.88564133644104\\
100	0.881081581115723\\
};
\addlegendentry{$\bar{\varepsilon} = 0.7$}

\addplot [color=mycolor5]
  table[row sep=crcr]{%
1	0.999996200203896\\
2	0.999996200203896\\
3	0.999996200203896\\
4	0.999996200203896\\
5	0.999996200203896\\
6	0.999996200203896\\
7	0.999996200203896\\
8	0.999996200203896\\
9	0.999996200203896\\
10	0.999996200203896\\
11	0.999996200203896\\
12	0.999996200203896\\
13	0.999996200203896\\
14	0.999996200203896\\
15	0.999996200203896\\
16	0.999996200203896\\
17	0.999996200203896\\
18	0.999996200203896\\
19	0.999996200203896\\
20	0.999996200203896\\
21	0.999996200203896\\
22	0.999996200203896\\
23	0.999996200203896\\
24	0.999996200203896\\
25	0.999996200203896\\
26	0.999996200203896\\
27	0.999996200203896\\
28	0.999996200203896\\
29	0.999996200203896\\
30	0.999996200203896\\
31	0.999996200203896\\
32	0.999996200203896\\
33	0.999996200203896\\
34	0.999996200203896\\
35	0.999996200203896\\
36	0.999996200203896\\
37	0.999996200203896\\
38	0.999996200203896\\
39	0.999996200203896\\
40	0.999996200203896\\
41	0.999996200203896\\
42	0.999996200203896\\
43	0.999996200203896\\
44	0.999996200203896\\
45	0.999996200203896\\
46	0.999996200203896\\
47	0.999996200203896\\
48	0.999996200203896\\
49	0.999996200203896\\
50	0.999996200203896\\
51	0.999996200203896\\
52	0.999996200203896\\
53	0.999996200203896\\
54	0.999996200203896\\
55	0.999996200203896\\
56	0.999996200203896\\
57	0.999996200203896\\
58	0.999996200203896\\
59	0.999996200203896\\
60	0.999996200203896\\
61	0.997963309288025\\
62	0.99525785446167\\
63	0.992096424102783\\
64	0.988631010055542\\
65	0.984861612319946\\
66	0.980849027633667\\
67	0.976654052734375\\
68	0.972246289253235\\
69	0.967716932296753\\
70	0.963005185127258\\
71	0.958171844482422\\
72	0.953247308731079\\
73	0.948261976242065\\
74	0.94315505027771\\
75	0.937987327575684\\
76	0.932789206504822\\
77	0.927530288696289\\
78	0.922240972518921\\
79	0.916921257972717\\
80	0.911601543426514\\
81	0.906251430511475\\
82	0.900901317596436\\
83	0.895551204681396\\
84	0.890201091766357\\
85	0.884850978851318\\
86	0.87956166267395\\
87	0.874272346496582\\
88	0.868983030319214\\
89	0.863693714141846\\
90	0.858465194702148\\
91	0.853297472000122\\
92	0.848129749298096\\
93	0.84302282333374\\
94	0.837915897369385\\
95	0.8328697681427\\
96	0.827884435653687\\
97	0.822899103164673\\
98	0.81797456741333\\
99	0.813110828399658\\
100	0.808247089385986\\
};
\addlegendentry{$\bar{\varepsilon} = 0.6$}

\addplot [color=mycolor6]
  table[row sep=crcr]{%
1	0.999996200203896\\
2	0.999996200203896\\
3	0.999996200203896\\
4	0.999996200203896\\
5	0.999996200203896\\
6	0.999996200203896\\
7	0.999996200203896\\
8	0.999996200203896\\
9	0.999996200203896\\
10	0.999996200203896\\
11	0.999996200203896\\
12	0.999996200203896\\
13	0.999996200203896\\
14	0.999996200203896\\
15	0.999996200203896\\
16	0.999996200203896\\
17	0.999996200203896\\
18	0.999996200203896\\
19	0.999996200203896\\
20	0.999996200203896\\
21	0.999996200203896\\
22	0.999996200203896\\
23	0.999996200203896\\
24	0.999996200203896\\
25	0.999996200203896\\
26	0.999996200203896\\
27	0.999996200203896\\
28	0.999996200203896\\
29	0.999996200203896\\
30	0.999996200203896\\
31	0.999996200203896\\
32	0.999996200203896\\
33	0.999996200203896\\
34	0.999996200203896\\
35	0.999996200203896\\
36	0.999996200203896\\
37	0.999996200203896\\
38	0.999996200203896\\
39	0.999996200203896\\
40	0.999996200203896\\
41	0.999996200203896\\
42	0.999996200203896\\
43	0.999996200203896\\
44	0.999996200203896\\
45	0.999996200203896\\
46	0.999996200203896\\
47	0.999996200203896\\
48	0.999996200203896\\
49	0.999996200203896\\
50	0.999996200203896\\
51	0.99746173620224\\
52	0.994041919708252\\
53	0.990059733390808\\
54	0.98563677072525\\
55	0.980849027633667\\
56	0.975802898406982\\
57	0.970452785491943\\
58	0.964920282363892\\
59	0.959144592285156\\
60	0.953247308731079\\
61	0.94722843170166\\
62	0.941087961196899\\
63	0.934886693954468\\
64	0.928563833236694\\
65	0.922240972518921\\
66	0.915857315063477\\
67	0.909473657608032\\
68	0.903029203414917\\
69	0.896584749221802\\
70	0.890201091766357\\
71	0.883817434310913\\
72	0.877433776855469\\
73	0.871050119400024\\
74	0.864788055419922\\
75	0.858465194702148\\
76	0.852263927459717\\
77	0.846062660217285\\
78	0.839982986450195\\
79	0.833903312683105\\
80	0.827884435653687\\
81	0.821865558624268\\
82	0.816029071807861\\
83	0.810192584991455\\
84	0.804356098175049\\
85	0.798641204833984\\
86	0.793047904968262\\
87	0.787454605102539\\
88	0.781922101974487\\
89	0.776450395584106\\
90	0.771039485931396\\
91	0.765750169754028\\
92	0.76046085357666\\
93	0.755232334136963\\
94	0.750125408172607\\
95	0.745018482208252\\
96	0.740033149719238\\
97	0.735108613967896\\
98	0.730184078216553\\
99	0.725381135940552\\
100	0.720638990402222\\
};
\addlegendentry{$\bar{\varepsilon} = 0.5$}

\addplot [color=mycolor7]
  table[row sep=crcr]{%
1	0.999996200203896\\
2	0.999996200203896\\
3	0.999996200203896\\
4	0.999996200203896\\
5	0.999996200203896\\
6	0.999996200203896\\
7	0.999996200203896\\
8	0.999996200203896\\
9	0.999996200203896\\
10	0.999996200203896\\
11	0.999996200203896\\
12	0.999996200203896\\
13	0.999996200203896\\
14	0.999996200203896\\
15	0.999996200203896\\
16	0.999996200203896\\
17	0.999996200203896\\
18	0.999996200203896\\
19	0.999996200203896\\
20	0.999996200203896\\
21	0.999996200203896\\
22	0.999996200203896\\
23	0.999996200203896\\
24	0.999996200203896\\
25	0.999996200203896\\
26	0.999996200203896\\
27	0.999996200203896\\
28	0.999996200203896\\
29	0.999996200203896\\
30	0.999996200203896\\
31	0.999996200203896\\
32	0.999996200203896\\
33	0.999996200203896\\
34	0.999996200203896\\
35	0.999996200203896\\
36	0.999996200203896\\
37	0.999996200203896\\
38	0.999996200203896\\
39	0.999996200203896\\
40	0.999996200203896\\
41	0.996671378612518\\
42	0.992096424102783\\
43	0.98677670955658\\
44	0.980849027633667\\
45	0.974465370178223\\
46	0.967716932296753\\
47	0.960603713989258\\
48	0.953247308731079\\
49	0.945708513259888\\
50	0.937987327575684\\
51	0.930144548416138\\
52	0.922240972518921\\
53	0.914276599884033\\
54	0.906251430511475\\
55	0.898226261138916\\
56	0.890201091766357\\
57	0.882175922393799\\
58	0.874272346496582\\
59	0.866368770599365\\
60	0.858465194702148\\
61	0.850683212280273\\
62	0.84302282333374\\
63	0.835362434387207\\
64	0.827884435653687\\
65	0.820406436920166\\
66	0.813110828399658\\
67	0.80581521987915\\
68	0.798641204833984\\
69	0.79158878326416\\
70	0.784657955169678\\
71	0.777848720550537\\
72	0.771039485931396\\
73	0.764412641525269\\
74	0.757846593856812\\
75	0.751402139663696\\
76	0.745018482208252\\
77	0.73881721496582\\
78	0.732615947723389\\
79	0.72659707069397\\
80	0.720638990402222\\
81	0.714741706848145\\
82	0.708966016769409\\
83	0.703311920166016\\
84	0.697718620300293\\
85	0.69212532043457\\
86	0.686775207519531\\
87	0.681425094604492\\
88	0.676135778427124\\
89	0.670968055725098\\
90	0.665861129760742\\
91	0.660875797271729\\
92	0.655890464782715\\
93	0.651026725769043\\
94	0.646284580230713\\
95	0.641542434692383\\
96	0.636921882629395\\
97	0.632301330566406\\
98	0.62780237197876\\
99	0.623303413391113\\
100	0.618926048278809\\
};
\addlegendentry{$\bar{\varepsilon} = 0.4$}

\addplot [color=mycolor8]
  table[row sep=crcr]{%
1	0.999996200203896\\
2	0.999996200203896\\
3	0.999996200203896\\
4	0.999996200203896\\
5	0.999996200203896\\
6	0.999996200203896\\
7	0.999996200203896\\
8	0.999996200203896\\
9	0.999996200203896\\
10	0.999996200203896\\
11	0.999996200203896\\
12	0.999996200203896\\
13	0.999996200203896\\
14	0.999996200203896\\
15	0.999996200203896\\
16	0.999996200203896\\
17	0.999996200203896\\
18	0.999996200203896\\
19	0.999996200203896\\
20	0.999996200203896\\
21	0.999996200203896\\
22	0.999996200203896\\
23	0.999996200203896\\
24	0.999996200203896\\
25	0.999996200203896\\
26	0.999996200203896\\
27	0.999996200203896\\
28	0.999996200203896\\
29	0.999996200203896\\
30	0.999996200203896\\
31	0.99525785446167\\
32	0.988631010055542\\
33	0.980849027633667\\
34	0.972246289253235\\
35	0.963005185127258\\
36	0.953247308731079\\
37	0.94315505027771\\
38	0.932789206504822\\
39	0.922240972518921\\
40	0.911601543426514\\
41	0.900901317596436\\
42	0.890201091766357\\
43	0.87956166267395\\
44	0.868983030319214\\
45	0.858465194702148\\
46	0.848129749298096\\
47	0.837915897369385\\
48	0.827884435653687\\
49	0.81797456741333\\
50	0.808247089385986\\
51	0.798641204833984\\
52	0.789278507232666\\
53	0.78009819984436\\
54	0.771039485931396\\
55	0.762223958969116\\
56	0.753530025482178\\
57	0.745018482208252\\
58	0.73675012588501\\
59	0.728603363037109\\
60	0.720638990402222\\
61	0.712796211242676\\
62	0.705135822296143\\
63	0.697718620300293\\
64	0.690362215042114\\
65	0.683127403259277\\
66	0.676135778427124\\
67	0.669265747070312\\
68	0.662517309188843\\
69	0.655890464782715\\
70	0.6494460105896\\
71	0.643123149871826\\
72	0.636921882629395\\
73	0.630842208862305\\
74	0.624823331832886\\
75	0.618926048278809\\
76	0.613211154937744\\
77	0.607617855072021\\
78	0.602146148681641\\
79	0.59667444229126\\
80	0.591385126113892\\
81	0.586217403411865\\
82	0.58111047744751\\
83	0.576064348220825\\
84	0.571139812469482\\
85	0.566276073455811\\
86	0.56153392791748\\
87	0.556913375854492\\
88	0.552292823791504\\
89	0.547793865203857\\
90	0.543416500091553\\
91	0.539039134979248\\
92	0.534783363342285\\
93	0.530649185180664\\
94	0.526515007019043\\
95	0.522502422332764\\
96	0.518489837646484\\
97	0.514598846435547\\
98	0.510707855224609\\
99	0.506938457489014\\
100	0.503169059753418\\
};
\addlegendentry{$\bar{\varepsilon} = 0.3$}

\addplot [color=mycolor9]
  table[row sep=crcr]{%
1	0.999996200203896\\
2	0.999996200203896\\
3	0.999996200203896\\
4	0.999996200203896\\
5	0.999996200203896\\
6	0.999996200203896\\
7	0.999996200203896\\
8	0.999996200203896\\
9	0.999996200203896\\
10	0.999996200203896\\
11	0.999996200203896\\
12	0.999996200203896\\
13	0.999996200203896\\
14	0.999996200203896\\
15	0.999996200203896\\
16	0.999996200203896\\
17	0.999996200203896\\
18	0.999996200203896\\
19	0.999996200203896\\
20	0.999996200203896\\
21	0.992096424102783\\
22	0.980849027633667\\
23	0.967716932296753\\
24	0.953247308731079\\
25	0.937987327575684\\
26	0.922240972518921\\
27	0.906251430511475\\
28	0.890201091766357\\
29	0.874272346496582\\
30	0.858465194702148\\
31	0.84302282333374\\
32	0.827884435653687\\
33	0.813110828399658\\
34	0.798641204833984\\
35	0.784657955169678\\
36	0.771039485931396\\
37	0.757846593856812\\
38	0.745018482208252\\
39	0.732615947723389\\
40	0.720638990402222\\
41	0.708966016769409\\
42	0.697718620300293\\
43	0.686775207519531\\
44	0.676135778427124\\
45	0.665861129760742\\
46	0.655890464782715\\
47	0.646284580230713\\
48	0.636921882629395\\
49	0.62780237197876\\
50	0.618926048278809\\
51	0.610414505004883\\
52	0.602146148681641\\
53	0.59399938583374\\
54	0.586217403411865\\
55	0.578557014465332\\
56	0.571139812469482\\
57	0.563965797424316\\
58	0.556913375854492\\
59	0.550104141235352\\
60	0.543416500091553\\
61	0.536972045898438\\
62	0.530649185180664\\
63	0.524447917938232\\
64	0.518489837646484\\
65	0.512653350830078\\
66	0.506938457489014\\
67	0.501345157623291\\
68	0.495995044708252\\
69	0.490644931793213\\
70	0.485416412353516\\
71	0.480431079864502\\
72	0.475445747375488\\
73	0.470582008361816\\
74	0.465961456298828\\
75	0.46134090423584\\
76	0.456841945648193\\
77	0.452342987060547\\
78	0.448087215423584\\
79	0.443831443786621\\
80	0.439697265625\\
81	0.435563087463379\\
82	0.431672096252441\\
83	0.427781105041504\\
84	0.423890113830566\\
85	0.420242309570312\\
86	0.416472911834717\\
87	0.412946701049805\\
88	0.409420490264893\\
89	0.40589427947998\\
90	0.40248966217041\\
91	0.399206638336182\\
92	0.395923614501953\\
93	0.392762184143066\\
94	0.38960075378418\\
95	0.386439323425293\\
96	0.38352108001709\\
97	0.380481243133545\\
98	0.377562999725342\\
99	0.374644756317139\\
100	0.371848106384277\\
};
\addlegendentry{$\bar{\varepsilon} = 0.2$}

\addplot [color=mycolor10]
  table[row sep=crcr]{%
1	0.999996200203896\\
2	0.999996200203896\\
3	0.999996200203896\\
4	0.999996200203896\\
5	0.999996200203896\\
6	0.999996200203896\\
7	0.999996200203896\\
8	0.999996200203896\\
9	0.999996200203896\\
10	0.999996200203896\\
11	0.980849027633667\\
12	0.953247308731079\\
13	0.922240972518921\\
14	0.890201091766357\\
15	0.858465194702148\\
16	0.827884435653687\\
17	0.798641204833984\\
18	0.771039485931396\\
19	0.745018482208252\\
20	0.720638990402222\\
21	0.697718620300293\\
22	0.676135778427124\\
23	0.655890464782715\\
24	0.636921882629395\\
25	0.618926048278809\\
26	0.602146148681641\\
27	0.586217403411865\\
28	0.571139812469482\\
29	0.556913375854492\\
30	0.543416500091553\\
31	0.530649185180664\\
32	0.518489837646484\\
33	0.506938457489014\\
34	0.495995044708252\\
35	0.485416412353516\\
36	0.475445747375488\\
37	0.465961456298828\\
38	0.456841945648193\\
39	0.448087215423584\\
40	0.439697265625\\
41	0.431672096252441\\
42	0.423890113830566\\
43	0.416472911834717\\
44	0.409420490264893\\
45	0.40248966217041\\
46	0.395923614501953\\
47	0.38960075378418\\
48	0.38352108001709\\
49	0.377562999725342\\
50	0.371848106384277\\
51	0.366254806518555\\
52	0.361026287078857\\
53	0.35579776763916\\
54	0.350812435150146\\
55	0.345948696136475\\
56	0.341328144073486\\
57	0.336707592010498\\
58	0.332330226898193\\
59	0.32807445526123\\
60	0.323940277099609\\
61	0.319806098937988\\
62	0.315915107727051\\
63	0.312145709991455\\
64	0.308376312255859\\
65	0.304728507995605\\
66	0.301323890686035\\
67	0.297919273376465\\
68	0.294514656066895\\
69	0.291231632232666\\
70	0.288070201873779\\
71	0.285030364990234\\
72	0.281990528106689\\
73	0.279072284698486\\
74	0.276154041290283\\
75	0.273357391357422\\
76	0.270682334899902\\
77	0.268007278442383\\
78	0.265332221984863\\
79	0.262778759002686\\
80	0.260225296020508\\
81	0.257793426513672\\
82	0.255361557006836\\
83	0.253051280975342\\
84	0.250741004943848\\
85	0.248552322387695\\
86	0.246363639831543\\
87	0.244174957275391\\
88	0.24210786819458\\
89	0.24004077911377\\
90	0.237973690032959\\
91	0.23602819442749\\
92	0.23396110534668\\
93	0.232137203216553\\
94	0.230313301086426\\
95	0.228367805480957\\
96	0.226665496826172\\
97	0.224841594696045\\
98	0.223017692565918\\
99	0.221315383911133\\
100	0.219613075256348\\
};
\addlegendentry{$\bar{\varepsilon} = 0.1$}

\addplot [color=mycolor11]
  table[row sep=crcr]{%
1	0.999996200203896\\
2	0.999996200203896\\
3	0.999996200203896\\
4	0.999996200203896\\
5	0.999996200203896\\
6	0.953247308731079\\
7	0.890201091766357\\
8	0.827884435653687\\
9	0.771039485931396\\
10	0.720638990402222\\
11	0.676135778427124\\
12	0.636921882629395\\
13	0.602146148681641\\
14	0.571139812469482\\
15	0.543416500091553\\
16	0.518489837646484\\
17	0.495995044708252\\
18	0.475445747375488\\
19	0.456841945648193\\
20	0.439697265625\\
21	0.423890113830566\\
22	0.409420490264893\\
23	0.395923614501953\\
24	0.38352108001709\\
25	0.371848106384277\\
26	0.361026287078857\\
27	0.350812435150146\\
28	0.341328144073486\\
29	0.332330226898193\\
30	0.323940277099609\\
31	0.315915107727051\\
32	0.308376312255859\\
33	0.301323890686035\\
34	0.294514656066895\\
35	0.288070201873779\\
36	0.281990528106689\\
37	0.276154041290283\\
38	0.270682334899902\\
39	0.265332221984863\\
40	0.260225296020508\\
41	0.255361557006836\\
42	0.250741004943848\\
43	0.246363639831543\\
44	0.24210786819458\\
45	0.237973690032959\\
46	0.23396110534668\\
47	0.230313301086426\\
48	0.226665496826172\\
49	0.223017692565918\\
50	0.219613075256348\\
51	0.216451644897461\\
52	0.213290214538574\\
53	0.210128784179688\\
54	0.207210540771484\\
55	0.204292297363281\\
56	0.201617240905762\\
57	0.198942184448242\\
58	0.196267127990723\\
59	0.193835258483887\\
60	0.191403388977051\\
61	0.188971519470215\\
62	0.186661243438721\\
63	0.184350967407227\\
64	0.182162284851074\\
65	0.180095195770264\\
66	0.178028106689453\\
67	0.176082611083984\\
68	0.174137115478516\\
69	0.172191619873047\\
70	0.170246124267578\\
71	0.168543815612793\\
72	0.166719913482666\\
73	0.165017604827881\\
74	0.163315296173096\\
75	0.161734580993652\\
76	0.160032272338867\\
77	0.158573150634766\\
78	0.156992435455322\\
79	0.155533313751221\\
80	0.154074192047119\\
81	0.152615070343018\\
82	0.151277542114258\\
83	0.149818420410156\\
84	0.148602485656738\\
85	0.147143363952637\\
86	0.145927429199219\\
87	0.144711494445801\\
88	0.143495559692383\\
89	0.142279624938965\\
90	0.141063690185547\\
91	0.139847755432129\\
92	0.138753414154053\\
93	0.137659072875977\\
94	0.1365647315979\\
95	0.135470390319824\\
96	0.13449764251709\\
97	0.133403301239014\\
98	0.132308959960938\\
99	0.131336212158203\\
100	0.130363464355469\\
};
\addlegendentry{$\bar{\varepsilon} = 0.05$}

\addplot [color=mycolor12]
  table[row sep=crcr]{%
1	0.999996200203896\\
2	0.720638990402222\\
3	0.543416500091553\\
4	0.439697265625\\
5	0.371848106384277\\
6	0.323940277099609\\
7	0.288070201873779\\
8	0.260225296020508\\
9	0.237973690032959\\
10	0.219613075256348\\
11	0.204292297363281\\
12	0.191403388977051\\
13	0.180095195770264\\
14	0.170246124267578\\
15	0.161734580993652\\
16	0.154074192047119\\
17	0.147143363952637\\
18	0.141063690185547\\
19	0.135470390319824\\
20	0.130363464355469\\
21	0.12574291229248\\
22	0.121608734130859\\
23	0.11759614944458\\
24	0.114069938659668\\
25	0.110665321350098\\
26	0.107503890991211\\
27	0.104585647583008\\
28	0.101910591125488\\
29	0.0992355346679688\\
30	0.0968036651611328\\
31	0.0946149826049805\\
32	0.0924263000488281\\
33	0.0903592109680176\\
34	0.0885353088378906\\
35	0.0865898132324219\\
36	0.0848875045776367\\
37	0.0831851959228516\\
38	0.08172607421875\\
39	0.0801453590393066\\
40	0.0788078308105469\\
41	0.0773487091064453\\
42	0.0761327743530273\\
43	0.0746736526489258\\
44	0.0735793113708496\\
45	0.0724849700927734\\
46	0.0712690353393555\\
47	0.0702962875366211\\
48	0.0692019462585449\\
49	0.0681076049804688\\
50	0.0671348571777344\\
51	0.0664052963256836\\
52	0.0654325485229492\\
53	0.0644598007202148\\
54	0.0637302398681641\\
55	0.0630006790161133\\
56	0.0620279312133789\\
57	0.0612983703613281\\
58	0.0605688095092773\\
59	0.0598392486572266\\
60	0.0593528747558594\\
61	0.0586233139038086\\
62	0.0578937530517578\\
63	0.0574073791503906\\
64	0.0566778182983398\\
65	0.0561914443969727\\
66	0.0554618835449219\\
67	0.0549755096435547\\
68	0.0544891357421875\\
69	0.0537595748901367\\
70	0.0532732009887695\\
71	0.0527868270874023\\
72	0.0523004531860352\\
73	0.051814079284668\\
74	0.0513277053833008\\
75	0.0508413314819336\\
76	0.0503549575805664\\
77	0.0501117706298828\\
78	0.0496253967285156\\
79	0.0491390228271484\\
80	0.0486526489257812\\
81	0.0484094619750977\\
82	0.0479230880737305\\
83	0.0476799011230469\\
84	0.0471935272216797\\
85	0.0467071533203125\\
86	0.0464639663696289\\
87	0.0459775924682617\\
88	0.0457344055175781\\
89	0.0454912185668945\\
90	0.0450048446655273\\
91	0.0447616577148438\\
92	0.0445184707641602\\
93	0.044032096862793\\
94	0.0437889099121094\\
95	0.0435457229614258\\
96	0.0433025360107422\\
97	0.042816162109375\\
98	0.0425729751586914\\
99	0.0423297882080078\\
100	0.0420866012573242\\
};
\addlegendentry{$\bar{\varepsilon} = 0.01$}

\addplot [color=mycolor13]
  table[row sep=crcr]{%
1	0.00585174560546875\\
2	0.00585174560546875\\
3	0.00585174560546875\\
4	0.00585174560546875\\
5	0.00585174560546875\\
6	0.00585174560546875\\
7	0.00585174560546875\\
8	0.00585174560546875\\
9	0.00585174560546875\\
10	0.00585174560546875\\
11	0.00585174560546875\\
12	0.00585174560546875\\
13	0.00585174560546875\\
14	0.00585174560546875\\
15	0.00585174560546875\\
16	0.00585174560546875\\
17	0.00585174560546875\\
18	0.00585174560546875\\
19	0.00585174560546875\\
20	0.00585174560546875\\
21	0.00585174560546875\\
22	0.00585174560546875\\
23	0.00585174560546875\\
24	0.00585174560546875\\
25	0.00585174560546875\\
26	0.00585174560546875\\
27	0.00585174560546875\\
28	0.00585174560546875\\
29	0.00585174560546875\\
30	0.00585174560546875\\
31	0.00585174560546875\\
32	0.00585174560546875\\
33	0.00585174560546875\\
34	0.00585174560546875\\
35	0.00585174560546875\\
36	0.00585174560546875\\
37	0.00585174560546875\\
38	0.00585174560546875\\
39	0.00585174560546875\\
40	0.00585174560546875\\
41	0.00585174560546875\\
42	0.00585174560546875\\
43	0.00585174560546875\\
44	0.00585174560546875\\
45	0.00585174560546875\\
46	0.00585174560546875\\
47	0.00585174560546875\\
48	0.00585174560546875\\
49	0.00585174560546875\\
50	0.00585174560546875\\
51	0.00585174560546875\\
52	0.00585174560546875\\
53	0.00585174560546875\\
54	0.00585174560546875\\
55	0.00585174560546875\\
56	0.00585174560546875\\
57	0.00585174560546875\\
58	0.00585174560546875\\
59	0.00585174560546875\\
60	0.00585174560546875\\
61	0.00585174560546875\\
62	0.00585174560546875\\
63	0.00585174560546875\\
64	0.00585174560546875\\
65	0.00585174560546875\\
66	0.00585174560546875\\
67	0.00585174560546875\\
68	0.00585174560546875\\
69	0.00585174560546875\\
70	0.00585174560546875\\
71	0.00585174560546875\\
72	0.00585174560546875\\
73	0.00585174560546875\\
74	0.00585174560546875\\
75	0.00585174560546875\\
76	0.00585174560546875\\
77	0.00585174560546875\\
78	0.00585174560546875\\
79	0.00585174560546875\\
80	0.00585174560546875\\
81	0.00585174560546875\\
82	0.00585174560546875\\
83	0.00585174560546875\\
84	0.00585174560546875\\
85	0.00585174560546875\\
86	0.00585174560546875\\
87	0.00585174560546875\\
88	0.00585174560546875\\
89	0.00585174560546875\\
90	0.00585174560546875\\
91	0.00585174560546875\\
92	0.00585174560546875\\
93	0.00585174560546875\\
94	0.00585174560546875\\
95	0.00585174560546875\\
96	0.00585174560546875\\
97	0.00585174560546875\\
98	0.00585174560546875\\
99	0.00585174560546875\\
100	0.00585174560546875\\
};
\addlegendentry{$\bar{\varepsilon} = 0$}

\end{axis}
\end{tikzpicture}%

%% file: Future_direction.tex
{We investigated private distributed matrix-vector multiplication schemes for sparse matrices that assign sparse computations to the {\workers} at the expense of guaranteeing weak information-theoretic privacy. We focused on the {\master}/{\worker} setting with two non-communicating clusters of {\workers}, where perfect privacy is required in one cluster and weak privacy in the other is satisfactory. For this setting we introduce a coding strategy that trades privacy for sparsity. By coupling the introduced coding strategy with fractional repetition codes we construct a sparse and private coded computing scheme that tolerates stragglers.}

{We are currently investigating the general {\master}/{\worker} setting with one cluster of {\workers}. The total leakage can be further decreased at a small loss in the average sparsity of the tasks. We suggest to choose $\pz \neq \pnz$ and study the resulting trade-off.}